\documentclass[10pt,draftcls,onecolumn]{IEEEtran}

\IEEEoverridecommandlockouts                              

\usepackage{empheq,amssymb}
\usepackage{braket}
\usepackage{dsfont,mathrsfs,bm}
\usepackage{enumerate}
\usepackage{float}
\usepackage{lipsum}
\usepackage{amsmath,amsthm}
\usepackage{amsthm}
\newtheorem{theorem}{Theorem}

\newtheorem{definition}{Definition}

\newtheorem{lemma}{Lemma}
\newtheorem{corollary}{Corollary}
\newtheorem{remark}{Remark}

\usepackage{mathtools}

\let \VEC \mathbf
\let \vec \mathbf
\def\ind{{\mathds{1}}}
\long\def\comment#1{}

\let \lessthan \prec
\let \morethan \succ

\newcounter{l1}
\newcounter{l2}
\newcounter{l3}
\newcommand{\bdotlist}{\begin{list}{$\bullet$}{}}
\newcommand{\bboxlist}{\begin{list}{$\Box$}{}}
\newcommand{\bbboxlist}{\begin{list}{\raisebox{.005in}{{\tiny $\blacksquare$ \ \ }}}{}}
\newcommand{\bdashlist}{\begin{list}{$-$}{} }
\newcommand{\blist}{\begin{list}{}{} }
\newcommand{\barablist}{\begin{list}{\arabic{l1}}{\usecounter{l1}}}
\newcommand{\balphlist}{\begin{list}{(\alph{l2})}{\usecounter{l2}}}
\newcommand{\bAlphlist}{\begin{list}{\Alph{l2}.}{\usecounter{l2}}}
\newcommand{\bdiamlist}{\begin{list}{$\diamond$}{}}
\newcommand{\bromalist}{\begin{list}{(\roman{l3})}{\usecounter{l3}}}

\usepackage{tikz}
\usepackage{rotating}
\usepackage{pgfplots}
\pgfplotsset{compat=newest}
\pgfplotsset{plot coordinates/math parser=false}

\newcommand{\la}{\langle}

\renewcommand{\ll}{\lambda}
\renewcommand{\hm}{h_{min}}
\newcommand{\expec}{\mathbb{E}}
\begin{document}

\title{
Duration-differentiated Energy Services with a Continuum of Loads}

\author{Ashutosh Nayyar,  Matias Negrete-Pincetic, Kameshwar Poolla and Pravin Varaiya
\thanks{The first author is with Ming Hsieh Department of Electrical Engineering at University of Southern California. The second author is with the Department of Electrical Engineering at Pontificia Universidad Catolica de Chile. The others are with the Department of Electrical Engineering and Computer Sciences, University of California, Berkeley. This work was supported in part by CERTS under sub-award 09-206;  NSF under Grants 1135872, EECS-1129061, CPS-1239178, CCF-1331692 and CNS-1239274.
}}

\maketitle
\thispagestyle{empty}
\pagestyle{plain}


\begin{abstract}
As the proportion of  total power supplied by  renewable sources increases, it gets more costly to use reserve generation to compensate for the  variability of  renewables like solar and wind. Hence  attention has been drawn to exploiting  flexibility in demand as a substitute for reserve generation.  Flexibility has different attributes.  In this paper we consider   loads  requiring a constant power  for a specified duration (within  say one day), whose flexibility resides in the fact that  power may be delivered  at any time  so long as the total duration of service equals the load's specified duration.  We  give conditions under which  a variable power supply is adequate to meet these flexible loads, and describe how to allocate the power to the loads.  We also characterize the  additional power needed when the supply is inadequate.   We study the problem of allocating the available power to loads to maximize welfare, and show that the welfare optimum can be sustained as a competitive equilibrium in a forward market in which electricity  is sold as service contracts differentiated by the duration of service and power level.  We compare this forward market with a spot market in their ability to capture the flexiblity inherent in duration-differentiated loads.
\end{abstract}

\section{Introduction} \label{sec-introduction}

Integration of renewable resources into the power grid poses  challenges. On the technological side, the uncertain and intermittent nature of renewable sources like solar and wind makes the control problem of balancing supply and demand more difficult.  On the economic side, it becomes necessary to offer incentives for consumers to shape their demand to match the  variable renewable supply.

Conventionally renewable power is treated as negative load, so variability in renewable power reappears as variability in net load, which is balanced using  reserve generation. But this approach becomes quite costly as renewable generation accounts for a significant portion of total power supply, see e.g.~\cite{CAISO, kirschen2010,negwankowshamey12}.  Consequently there is growing interest in schemes that exploit the flexibility in electricity consumption  to compensate for the variability in supply.  Several such schemes have been studied, see e.g. ~\cite{cal09K, galus2010, papaoren2010, matdyscal12K, anand2012, Paschalidis}.  There is unlikely to be a single best scheme since loads vary in their attributes of flexibility:  a  load  requiring a fixed quantity of energy over a fixed period with no restriction as to when and at what rate this energy is delivered is very flexible; typically, however,  loads  have constraints on the rates of energy consumption, power limits, interruptions, etc that must be taken into account in the characterization of flexible loads. 

Consumers may need  compensation in exchange for permitting suppliers to exploit  their flexibile loads, and this may require  economic incentives that go beyond those inherent in a  real time price for energy,   e.g. ~\cite{tanvar93, pravin2011,negmey12, bitarlow2012}.  New commodities or service contracts may need to be designed to elicit consumer response tailored to the attributes of renewables and flexible loads. Examples of such new electricity services include service reliability and deadline-differentiated contracts  \cite{tanvar93, bitarlow2012}.
 
In this paper, we consider  loads that require a fixed power level for a specified duration.  Their flexibility resides in the fact that the time when service is delivered is immaterial so long as the total service duration (within say one day) equals the load-specified duration.  So loads are differentiated by the duration of service  and the power level.  We call them \textit{duration-differentiated} or DD loads The supplier is free to schedule DD loads within their constraints
so as to match the variability in supply. 

  We  characterize when a given power supply profile  is adequate to meet a set of DD loads, and describe an algorithm that schedules the loads when the supply is adequate. We  determine the additional power needed when the supply is inadequate.  We find the  allocation of power that maximizes the aggregate utility of a continuum of consumers, with no restriction on the shape of their utility functions.  Furthermore we show that this welfare maximizing allocation is an equilibrium in a competitive market for DD services whose  energy prices  vary by  service duration.  We compare this market with a day ahead and a real time or spot market in terms of their ability to capture the flexibility of DD loads.

We consider  a continuum of loads as in  \cite{tanvar93}, that is, each load is tiny compared with the aggregate load.  (Think of individual loads in the kW range and the aggregate load in the MW range.)  The  continuum of loads allows us to avoid bin-packing problems in scheduling loads and to avoid imposing  concavity conditions on the utility functions.

The rest of the paper is organized as follows. In \S \ref{sec-adequacy} we investigate supply adequacy, load scheduling and supplemental power. In \S \ref{sec-market},  we analyze a forward market for duration differentiated services. 
We consider two illustrative examples of this forward market for a collection of consumers with the same  utility function in \S \ref{sec-illustra}.  We compare this market with a spot market in \S \ref{sec-compare}.  Concluding remarks and future lines of research are provided in \S \ref{sec-conclusions}.

\subsection*{Notation}

Boldface letters denote vectors.  Time is slotted and $t=1, \cdots, T$ denotes a time slot, with $T$ slots equal to (say) one day.   For an assertion $A$, $\mathds{1}_{A}$ is the indicator function that equals $1$ if $A$ is true and $0$ if $A$ is false. 






\section{Adequacy Conditions}\label{sec-adequacy}

We  find the conditions under which  the available power profile is adequate to satisfy a set of loads.
The  conditions are expressed  in terms of \textit{majorization} relationships between supply and demand. 
We use the following terminology.


\textbf{Consumers} There is a continuum of consumers indexed by $x \in [0,1]$. 
 Consumer $x$ demands power `slice' $(\ell(x), h(x))$ representing $\ell(x) \ge 0$ kW of power for a duration of $h(x) \in \{1,2,\ldots,T\}$ slots. Interpret $\ell(x)$ as per capita demand, so  consumers in $[x,x+dx]$ demand  $\ell(x)dx$ kW of power for $h(x)$ slots.  These consumers are indifferent as to when electricity is delivered so long as  the duration of service is $h(x)$ slots. $\ell(\cdot)$ and $h(\cdot)$ are  measurable functions. 
The terms `consumer' and `load' are used interchangeably.

\textbf{Supply} 
The  power available in slot $t$ is $p_t$.  Reorder the time axis so that $p_1 \ge p_2 \ge \cdots \ge p_T$.  Call the vector $\VEC p=(p_1,p_2,\ldots,p_T)$ the \textit{supply profile}. It is also the generation duration curve.  When we need to refer to the power supply in the original time order we will call it
the \textit{supply time profile} and use the notation $\VEC q = (q_1, \cdots , q_T)$.  $\VEC p$ is a permutation of $\VEC q$.

\textbf{Demand} The (aggregate) \emph{demand profile} is the vector $ \VEC d=(d_1,d_2,\ldots,d_T)$ with 
\[d_t=\int_0^1 \ell(x) \ind(h(x) \geq t ) dx.\]  
So $d_t$ is the energy need of  loads requiring at least $t$ slots of service.  Hence $d_1 \ge d_2 \ge \cdots \ge d_T$.  The demand profile is also the demand duration curve.

\subsection{Adequacy}
\begin{definition}
The supply profile $\VEC p=(p_1,p_2,\ldots,p_T)$ is  \emph{exactly adequate} for loads $(\ell(x),h(x))$, $x \in [0,1]$, if there is an allocation  $A(x,t) \in \{0,1\}$, $x \in [0,1]$, $t \in \{1,2,...,T\}$  so that
\begin{align}
\int_0^1 \ell(x)&A(x,t)dx = p_t ,\label{ca1}\\
\sum_{t=1}^T &A(x,t) = h(x) .\label{ca2}
\end{align}
\end{definition}
\noindent 
Thus
$A(x,t) = 1$ if and only if $x$ receives electricity in slot $t$;
 (\ref{ca1}) states that  the available power is used up;    (\ref{ca1}) and  (\ref{ca2}) together imply that the requirement of every load is met.

The next result is proved similarly to  the case of rate-constrained services with finitely many consumers studied in \cite{dd2013}.
\begin{theorem}\label{thm:cont_adequacy}
The supply profile  $\VEC p=(p_1,p_2,\ldots,p_T)$ is exactly adequate for the demand profile  $ \VEC d=(d_1,d_2,\ldots,d_T)$ associated with demands $(\ell(x),h(x))$  if and only if $\VEC p$ is majorized by $\VEC d$ ($\VEC p \morethan \VEC d$) i.e., 
\begin{align}
\sum_{t=1}^s p_t &\leq \sum_{t=1}^s d_t,\;\;\;s=1,...,T-1 ,\label{a1}\\
\sum_{t=1}^T p_t &= \sum_{t=1}^T d_t . \label{a2}
\end{align}
\end{theorem}
\begin{proof}
We first prove sufficiency.
Since $d \morethan d$ and $d$ is exactly adequate for $d$, so to prove sufficiency, it is enough to show  that if $\VEC a$ is exactly adequate, any  profile with $\VEC b \morethan \VEC a$ is also exactly adequate.  We first recall a  lemma from majorization theory.
\begin{definition}
Define a \textit{Robin Hood (RH)} transfer on $\VEC a$ to be an operation that
\begin{enumerate}[(i)]
\item Selects indices $t,~s$ such that $a_t > a_s$,
\item Replaces $a_t$ by $a_t - \epsilon$ and $a_s$ by $a_s+\epsilon$, for any $\epsilon \in (0,a_t-a_s)$,
\item Rearranges the new vector in a non-increasing order.
\end{enumerate}
\end{definition}
\begin{lemma}  \cite{arnold}
Let $\VEC b \morethan \VEC a$ with $\VEC a \neq \VEC b$. Then, there exists a sequence of RH transfers that can be applied on $\VEC a$ to get $\VEC b$. 
\end{lemma}

It remains to prove that a RH transfer preserves exact adequacy. Let $\VEC a$  be exactly adequate  with allocation function $A(x,t)$. By applying a RH transfer of amount $\epsilon$ from slot $t$ to $s$ a new profile $\VEC{\tilde{a}}$ is obtained. 

Since $\int_0^1 \ell(x) [A(x,t)-A(x,s)] dx = a_t - a_s > \epsilon$, there exists $\Lambda \subset [0,1]$ such that 
\[ A(x,t) =1, \ A(x,s) = 0, x \in \Lambda , \mbox{ and } \int_{\Lambda} A(x,t) l(x) dx = \epsilon.\]
Then the profile $\VEC{\tilde{a}}$ is also exactly adequate for the allocation $\tilde{A}$ in which the loads $x \in \Lambda$ are deferred from $t$ to $s$:
 \begin{equation} 
\tilde{A}(x,r) = \left\{
\begin{array}{ll}
1 & r=s, ~x\in \Lambda\\
0 & r=t,~ x \in \Lambda \\
A(x,r) & \mathrm{otherwise} 
\end{array}
\right. .
\end{equation}
To prove the necessity, note that if the allocation  $A(x,t)$ satisfies \eqref{ca2}, then it must be true that for $s \geq 1$
\begin{align}
\sum_{t=s}^T A(x,t) &\geq (h(x)-s+1)_+ 
=\sum_{t=s}^T \ind_{\{h(x) \geq t\}} .\label{eq:nec_proof1}
\end{align}
Multiply  \eqref{eq:nec_proof1} by $\ell(x)$ on both sides, integrate and sum
\begin{align}
 & \sum_{t=s}^T \int_0^1 \ell(x)A(x,t)dx \geq \sum_{t=s}^T\int_0^1 \ell(x)\ind(h(x) \geq t)dx  \notag \\
 &\implies \sum_{t=s}^T p_t \geq \sum_{t=s}^T d_t, \label{eq:cont_necc2}
\end{align}
where we used \eqref{ca1} and the definition of $d_t$ in \eqref{eq:cont_necc2}. For $s=1$, the inequalities in \eqref{eq:nec_proof1}, \eqref{eq:cont_necc2} become equalities. 
\end{proof}

 The inequalities in our definition of majorization are reversed from standard usage to allows us to read the adequacy condition
 $\VEC d \lessthan \VEC p$  as demand being `less than' supply.
It is worth noting that \eqref{eq:cont_necc2} gives the adequacy condition in terms of  the `tail' energy in the demand and supply profiles:
\begin{corollary}\label{cor:cont_adequacy}
The condition for adequacy $\VEC p \morethan \VEC d$ can be written as
\begin{align}
\sum_{t=s}^T d_t &\leq \sum_{t=s}^T p_t,\;\;\;s=2,...,T\\
\sum_{t=1}^T p_t &= \sum_{t=1}^T d_t
\end{align}
\end{corollary}

\begin{corollary}\label{cor:reserve}
Suppose $\VEC p = (p_1, \cdots , p_T)$, with $p_t = E/T$ constant and total energy $E$.  Let $\VEC d = (d_1, \cdots, d_T)$ be any demand profile requiring total energy $\sum p_t = E$.  Then $p$ is adequate for $d$.
\end{corollary}
\begin{proof}
$D(s) = \sum _{t=1}^s d_s$ is a `concave' function since $d_s$ is non-increasing.  $P(s) = \sum_{t=1}^s p_s = (E/T)s$ is `linear'.  Since $D(0) = P(0) = 0$
and $D(T) = P(T) =E$, it follows that  $P(s) \le D(s)$ for all $s$, so \eqref{a1} and \eqref{a2} are satisfied and $p$ is adequate for $d$. 
\end{proof}
\begin{remark}
Corollary \ref{cor:reserve}  provides a simple illustration of how demand flexibility can substitute for reserve capacity.  As above suppose $(l(x), h(x))$ is the 
flexible demand of consumer $x \in [0,1]$, leading to the demand profile $(d_1, \cdots, d_T)$ with total energy demand $E = \sum d_t$.  By Corollary \ref{cor:reserve} this demand can be met
by the constant supply profile $p_t= E/T$.

Now suppose that consumers are required to express their demand for each slot in a day or $T$ slot-ahead market.  Then  $x$'s demand will be of the form
$\ell (x) $kW for slots $t \in H(x)$ where  $H(x) = \{t_1, \cdots , t_{h(x)}\}$ comprises the $h(x)$ distinct slots that $x$ has selected.  In this day-ahead market the aggregate demand for slot $t$ will be 
\[\delta_t = \int_0^1 \ell (x) \ind (t \in H(x)) dx .\]
To meet this demand, the supplier will need generation with average power $p_{av} =E/T$ but which  can also supply the peak demand $p_{peak} = \max_t \delta_t$.   The  ratio 
\[R = [p_{peak} - p_{av}]/p_{av}\]
 is a measure of the reserve capacity that is needed to meet this day-ahead demand.  If, however, the
supplier can exploit the demand flexibility, $R=0$ and \textit{ no }reserve capacity would be needed.
\end{remark}

\subsection{Simple Adequacy}
A  less strict notion of adequacy is also useful.
\begin{definition}
The supply profile $\VEC p=(p_1,p_2,\ldots,p_T)$ is  \emph{simply adequate}  for loads $(\ell(x),h(x)), x \in [0,1]$ if there is an allocation  $A(x,t) \in \{0,1\}$, $x \in [0,1]$, $t \in \{1,2,...,T\}$  such that
\begin{align}
\int_0^1 \ell(x)&A(x,t)dx \leq p_t ,\label{cas1}\\
\sum_{t=1}^T &A(x,t) = h(x) .\label{cas2}
\end{align}
\end{definition}
From \eqref{cas1}, \eqref{cas2} it follows that
\begin{equation}
\sum_{t=1}^T p_t    \geq  \int_0^1 \ell(x)h(x)dx = \sum_{t=1}^T d_t,
\end{equation}
so the total available energy supply  exceeds the total  energy required by the loads.
Simple adequacy can also be stated in terms of  the tail energy.
\begin{theorem} \label{thm2}
 $\VEC p$ is simply adequate if and only if
  \begin{align}
\sum_{t=s}^T d_t &\leq \sum_{t=s}^T p_t,\;\;\;s=1, \cdots ,T. \label{sa}
\end{align}
\end{theorem}
\subsection{Allocation}
Suppose $\VEC p$ is simply adequate for $\VEC d$. Let $\VEC{q} = (q_1, \cdots, q_T)$ be the supply time profile.  So $\VEC q$ is a permutation of $\VEC p$.   The following theorem describes a (causal) allocation algorithm.
%
%
\begin{theorem}\label{thm:continuum_allocation}
Suppose $\VEC p$ is simply adequate for  loads $(\ell(x),h(x)), x \in [0,1]$.  Let $\VEC q$ be the supply time profile corresponding to   $\VEC p$. Construct the  allocation $A(x,t)$:
\begin{enumerate}
\item At slot 1, define $y_1(x)=h(x)$. Find the smallest non-negative integer $k$ such that
\begin{equation}\nonumber
\int_0^1 \ell(x) \ind_{\{y_1(x) \geq k\}}dx \leq q_1.
\end{equation}
Pick the largest $x' \in [0,1]$ such that
\begin{equation}\nonumber
\int_0^1 \ell(x) \ind_{\{y_1(x) \geq k\}}dx + \int_0^{x'} \ell(x) \ind_{\{y_1(x) = k-1\}}dx \leq q_1 ,
\end{equation}
let  $\mathcal{A}_1=\{x : y_1(x) \geq k \} \cup \left([0,x'] \cap \{x : y_1(x) =k-1 \}  \right)$, and set $A(x,1)=1 \iff x \in \mathcal{A}_1$.
\item At slot $t$, define $y_t(x)=h(x) - \sum_{s=1}^{t-1} A(x,s)$. Find the smallest non-negative integer $k$ such that
\begin{equation}\nonumber
\int_0^1 \ell(x) \ind_{\{y_t(x) \geq k\}} \leq q_t .
\end{equation}
Pick the largest $x' \in [0,1]$ such that
\begin{equation}\nonumber
\int_0^1 \ell(x) \ind_{\{y_t(x) \geq k\}}dx + \int_0^{x'} \ell(x) \ind_{\{y_t(x) = k-1\}}dx \leq q_t ,
\end{equation}
let  $\mathcal{A}_t=\{x : y_t(x) \geq k \} \cup \left([0,x'] \cap \{x : y_t(x) =k-1 \}  \right)$, and set $A(x,t)=1 \iff x \in \mathcal{A}_t$.
\end{enumerate}
The allocation $A$ satisfies all the demands.
\end{theorem}
\begin{proof}
See Appendix.
\end{proof}
\begin{remark}
We call this allocation the Longest Leftover Duration First (LLDF) rule because at each slot $t$ it gives higher priority to loads with a longer leftover duration $y_t(x) = h(x) - \sum_{s=1}^{t-1} A(x,s)$.
$\VEC q$ is the supply time profile, this rule is on-line. 
 Before slot 1, the supplier must know the load requirements $(\ell (x), h(x))$, $x \in [0,1]$, but
then before slot $t$, the supplier only needs to know the current supply $q_t$ to decide which loads to schedule for slot $t$.  If the supply is adequate, all loads will be served.  If the supply is not adequate additional power must be procured, as discussed next.
\end{remark}

\subsection{Additional Power Procurement}
If the supply profile $\VEC p$ is not adequate, then the supplier may have to purchase additional supply profile $\VEC a = (a_1,\ldots,a_T)$ in order to serve the loads. 
With a linear cost for  additional power, the  power to be purchased at minimum cost while ensuring that all demands are met is given by the solution of the following optimization problem:
\begin{equation}
    \min_{\VEC a \geq 0} \sum_{t=1}^T c \cdot a_t \notag \quad
    \mbox{subject to~~~} \VEC p + \VEC a \morethan \VEC d, \label{LP}
\end{equation}
where $c \geq 0$ is the unit price of additional power.  The minimum value of this optimization problem and an algorithm for finding the minimizing vector $\VEC a$ has been provided in \cite{dd2013}.

%
%


%
%

\section{Market Implementation}\label{sec-market}

We investigate a forward market for duration-differentiated energy services in which  all market transactions are completed at time 0, before slot 1. For the moment we do not consider any uncertainty.  The market has three elements:

\begin{itemize}

\item Services: The services $(l,h)$ are differentiated by the number $h$ of time slots during which $\ell$ kW  is delivered.  Service $(\ell, h)$ is sold at price $\pi_h \times \ell$.

\item Consumers:  The benefit to consumer $x \in [0,1]$ who receives service $(\ell, h)$ is  $U(x,\ell,h)$.  $U(x, \ell, h) \ge 0$ is a bounded measurable function with $U(x,0,0) = 0$.  There is no convexity assumption on $U$.  The net benefit to consumer $x$ is $U(x, \ell, h) - \pi_h \times \ell$.

\item Supplier: The (aggregate) supplier knows the  supply profile $\VEC p = (p_1, \cdots, p_T)$ at time 0. 

\end{itemize}

%

We first formulate the  social welfare maximization problem and then show that the optimum  allocation is sustained as a competitive equilibrium.

\subsection{Social Welfare Problem}

For a given supply profile $\VEC p = (p_1,\ldots,p_T)$ say that
an allocation  $x \mapsto (\ell(x),h(x))$ is \textit{feasible} if $\VEC p$ is simply adequate for the associated demand profile. By Theorem \ref{thm2} the social welfare optimization problem is
\begin{align}
    &\max_{\ell(x),h(x)} \int_0^1 U(x,\ell(x),h(x))dx \label{16}\\
    & \mbox{subject to} \notag \\
    &\sum_{i=t}^T d_i \leq \sum_{i=t}^T p_i, \quad 
    d_i = \int_0^1 \ell(x)\ind{\{i \leq h(x)\}}dx, \quad i=1,2,\ldots,T, \notag \\
    & \ell(x) \geq 0,\;\;h(x) \in \{1,2,\ldots,T\}. \notag
\end{align}
Define $\VEC z(x) = (z_0(x),z_1(x),\ldots,z_T(x))$ with $\VEC z(0)=0$ and 
\begin{eqnarray}
\dot{z_0}(x) &=& U(x,\ell(x),h(x)) , \label{eq:de_1}\\
\dot{z_t}(x)&=&\ell(x)\sum_{i=t}^T\ind{\{i \leq h(x)\}}, \notag\\
&=& \ell(x) [h(x)+1 -t]_+ , \ t=1, \cdots, T .\label{eq:de_2}
\end{eqnarray}
Then
\begin{eqnarray}
z_0(1) &=& \int_0^1 U(x,\ell(x),h(x))dx ,\notag \\
z_t(1) &=&  \sum_{i=t}^T \int_0^1 \ell(x)\ind{\{i \leq h(x)\}}dx.
\end{eqnarray}
Thus, $z_0(1)$ is the value of the objective function and  $z_t(1) \leq \sum_{i=t}^T p_i$, $t=1,2,\ldots,T$ are the constraints of the  welfare maximization problem.
 
For $0 \le x \le 1$, define $F(x) \subset R^{1+T}$ by
\begin{equation}
\begin{aligned}
F(x) = \Big\{U(x,\ell(x),h(x)), \ell(x)h(x),  \ell(x) [h(x)-1]_+ \\
\ldots \ell(x) [h(x)+1 -T]_+  ~\mid~ \ell(x) \ge 0, h(x) \in \{1, \cdots, T\}\Big\} .\nonumber
\end{aligned}
\end{equation}
$x \mapsto F(x) \subset R^{1+T}$ is a set-valued function. 
The differential equations \eqref{eq:de_1} and \eqref{eq:de_2} can be
written as the differential inclusion,
\[\dot{\VEC z}(x) \in F(x).\]
 The integral of the set-valued function $F(x)$ is  the set
 \begin{align} 
 &G = \Big\{ \int_0^1 f(x)dx \Big| \begin{array}{l}f~\mbox{ is a measurable function with}\\
 f(x) \in F(x) \end{array} \Big\}.
 \end{align}
\begin{theorem}
$G $ is convex and closed.  
\end{theorem}
\begin{proof}
The proof relies on a theorem of Lyapunov on the convexity of the range of a vector
 valued integral \cite{aumann65}.
\end{proof}
Observe that
\begin{equation*}
G = \left\{\VEC z(1) ~\Big|~ \begin{array}{l}\VEC z(1) \mbox{ is reached by a feasible allocation}~\\  x \mapsto (\ell(x), h(x)) \end{array}\right\}.
\end{equation*}  

The welfare maximization problem restated in terms of $G$ is
\begin{eqnarray}
J = \max_{\VEC z(1)} && z_0(1) \label{j1} \\
\mbox{s.t.} && z_t (1) \le \sum_{i=t}^T p_i, \ t=1, \cdots, T,\label{3}\\
&& \VEC z(1) \in G .\label{4}
\end{eqnarray}

\begin{theorem} \label{theorem5}
The welfare maximization problem has a solution for any bounded measurable utility function $U(x,l,h)$.
\end{theorem}
\begin{proof}
The set of $\VEC z(1)$ satisfying \eqref{3}, \eqref{4} is compact and convex.  So the optimum $\VEC z^*(1)$ exists.  
\end{proof}
Theorem \ref{theorem5}  is a  consequence of assuming a continuum of users. Similar results are available for example in the study of an economy with a continuum of traders \cite{aumann2}.

\begin{remark}{\rm 
In standard commodity markets, one often considers concave utility functions to reflect decreasing marginal utility of consumption. In the case of duration-differentiated loads, the concave case is illustrated by the decreasing marginal utility  of the number of hours a pool is heated or a room is cooled.  However,
if a consumer only wants a minimum number of hours of pool heating or air cooling, the resulting utility function is
$U(x, \ell, h) = \ind [\ell \ge \ell_0, h \ge h_0]$, which  is not concave.}
\end{remark}

\subsection{Competitive Equilibrium Analysis}
We study a competitive equilibrium for the market for duration-differentiated services.   Recall the definition of a competitive equilibrium.

\begin{definition}
For the services  of durations $h=1,2,\ldots,T$ and associated prices $\pi_h$ per kW,  a competitive equilibrium requires three conditions:

\begin{itemize}
\item Consumers maximize their welfare, that is,
\begin{equation}
(\ell (x), h(x)) \in \arg
\max_{\ell,h} U(x,\ell,h) - \pi_h \ell .
\end{equation}

\item Supplier maximizes  revenue.  
The supplier uses the supply profile $\VEC p$ to offer a bundle of services $\{L_h, h\}, h=1, \cdots, T$.  The bundle must be feasible, i.e., 
\begin{align}
\sum_{i=t}^T \delta_i \leq \sum_{i=t}^T p_i, \quad t=1,\ldots,T \notag,
\end{align}
where $\delta_i :=  \sum_{s=i}^T L_s $ is the demand profile associated with the bundle of services $L_1,L_2,\ldots,L_t$.

 The supplier's profit maximization problem is to choose a feasible production bundle to maximize its revenue:
\begin{equation}
\begin{aligned}
\max& \sum_{h=1}^T L_h \pi_h\\
&s.t.\;\;\; \\
&\sum_{i=t}^T \delta_i \leq \sum_{i=t}^T p_i, \quad t=1,\ldots,T \notag,
\end{aligned}
\end{equation}
where $\delta_i :=  \sum_{s=i}^T L_s $.
\item The market clears, that is,
\begin{equation}
L_h = \int_0^1 \ell(x) \ind_{\{h(x)=h\}}.
\end{equation}
    \end{itemize}
\end{definition}
A competitive equilibrium is called \textit{efficient} if the resulting allocation maximizes social welfare.
\begin{theorem}
There exists an efficient competitive equilibrium in a forward market for duration-differentiated services in which the service of duration $h$ is traded at a  price of $\pi_h$ per kW. 
\label{thm:market}
\end{theorem}
\begin{proof}
Dualizing the social welfare problem with respect to \eqref{3} implies that there 
exist non-negative numbers $\lambda_1, \ldots \lambda_T$ such that the optimum $\VEC z^*(1)$ maximizes
\begin{eqnarray}
J= \max_{\VEC z(1) \in G}  z_0 (1) - \sum_{t=1}^T \lambda_t z_t(1) ,\label{6}
\end{eqnarray}
and satisfies the complementary slackness conditions,
\begin{equation}
\lambda_t\Big(z_t(1) - \sum_{i=t}^T p_i\Big) = 0, ~~~t=1,\ldots,T .\label{eq:comp_slck}
\end{equation}

Given the Lagrange multipliers $\lambda_1,\ldots,\lambda_T$, the term being maximized in \eqref{6} can be written as
\begin{equation}
\int_0^1 \Big[ U(x,\ell(x),h(x) - \ell(x)\sum_{t=1}^T\lambda_t[h(x)+1-t]_+ \Big] dx .\label{eq:neweq1}
\end{equation}
In order to maximize the value of the integral in \eqref{eq:neweq1}, 
\begin{align}
(\ell^*(x),h^*(x)) &= \arg\max_{\ell,h} U(x,\ell,h) -
\ell\sum_{t=1}^T\lambda_t[h+1-t]_+ \label{eq:neweq2}
\end{align}
Interpreting the quantity $\sum_{t=1}^T\lambda_t[h+1-t]_+$ as the per kW price $\pi_h$ for duration $h$,  \eqref{eq:neweq2} amounts to the consumer welfare maximization condition of equilibrium. 

In order to prove that $\pi_h = \sum_{t=1}^T\lambda_t[h+1-t]_+ $ are indeed equilibrium prices, we need to show that the supplier revenue maximization at these prices will result in market clearing. 

The total revenue of the supplier can be written as 
\begin{align}
\sum_{h=1}^T L_h\pi_h = \sum_{k=1}^T \delta_k(\pi_k -\pi_{k-1}) \label{eq:neweq3}
\end{align}
where $\delta_k = \sum_{t \geq k} L_t$. Further,  using the expression for $\pi_k$, \eqref{eq:neweq3} can be written as
\begin{align}
&\sum_{k=1}^T \delta_k \left(\sum_{i=1}^k \lambda_i\right)
= \sum_{i=1}^T \lambda_i \left(\sum_{k=i}^T \delta_k\right) 
\leq \sum_{i=1}^T \lambda_i \left(\sum_{k=i}^T p_k\right),\label{eq:neweq4}
\end{align}
where we used the constraints from supplier's optimization problem in \eqref{eq:neweq4}. Because $\vec z^*(1)$ satisfies the complementary slackness conditions  \eqref{eq:comp_slck}, the revenue in \eqref{eq:neweq4} is bounded from above by
\[ \sum_{i=1}^T \lambda_i z^*_i(1).\]
It is easy to verify that this upper bound on revenue is achieved by  $L_h = \int_0^1 \ell^*(x)\ind_{\{h^*(x) = h\}}$. Thus, the supplier's revenue maximization results in the correct bundle of  services needed to clear the market. 
\end{proof}
\begin{remark}
It follows from the definition of prices $\pi_h =\sum_{t=1}^T\lambda_t[h+1-t]_+ $ that $\pi_h$ is non-decreasing in $h$ with non-decreasing increments. In particular,
\[\pi_h-\pi_{h-1} = \lambda_1 + \lambda_2+\ldots + \lambda_h.\]
 This reflects the fact that it is more difficult to provide  a service of duration $h_1 + h_2$ than
two services of duration $h_1$ and $h_2$.    This may induce consumers 
 to install devices, e.g., storage, that can bridge the service interruptions caused by purchasing two services of durations $h_1$ and $h_2$ rather than a single service of duration $h_1 + h_2$.
\end{remark}

\begin{remark}
\label{remarkslot}
Define 
\begin{equation}
\mu_t=\sum_{i=1}^t \lambda_i^* . \label{32}
\end{equation}
Since the total supply in slot $t$ is $s_t = \sum_{i=t}^T p_i$, we see that $ \la \boldsymbol \ll^*, \boldsymbol s\rangle = \la \boldsymbol \mu^*, \boldsymbol p \rangle$.  $\mu_t$ can be viewed as the per kW price  of each time slot. Suppose the supply time profile is $\VEC q$ is the permutation of $\VEC p$ given by $q_t(i) = p_i$.  Then $q(t(1)) \ge \cdots \ge 
q(t(T)))$.  So $\mu_1$ is the per kW price in slot $t(1)$ where the supply is largest, $\mu_2$ is the price in slot $t(2)$ where the supply is second largest, and so on.  As expected, \eqref{32} implies $\mu_1 \le \mu_2 \cdots \leq \mu_T$.
\end{remark}

\begin{remark}
%
Suppose consumers classified some consumption as DD load  requiring $h$ slots during a  period of $T$ slots, allowing the supplier to
select the $h$ slots.  These customers are similar to very reliable demand-response subscribers.  Instead of being paid for reducing their consumption in response to a `demand response' event, they would get a discount for a DD load.
\end{remark}
\begin{remark}{ In our formulation the supply side represents an aggregation of many suppliers.  Individual suppliers may benefit by pooling their supplies to offer longer service contracts and take advantage of higher prices.  Thus we are assuming that  suppliers are able to identify and exploit all such opportunities.  We implicitly assume that even with such coordination the number of effective suppliers is large enough to justify the assumption of a competitive market. }
\end{remark}

\section{Illustrative Cases} \label{sec-illustra}


We construct the equilibrium  contracts and prices when all consumers have the same utility function $U(\ell,h)$ with $U(\ell,h) =0$ if either $\ell =0$ or $h=0$. We consider  two cases.

\subsection{$U(\ell,h)$ is strictly concave in $\ell$ for all $h>0$ and has non-decreasing, strictly positive increments in $h$ for all $\ell >0$.}\label{sec:4a}

Since consumers are identical, in equilibrium they must receive the same net benefit or consumer surplus $H$, though not necessarily the same allocation.  The equilibrium prices and allocations  $\boldsymbol \lambda^*, \ell^*,h^*$ are parameterized by $H$.  The approach is similar to that in  \cite{tanvar93}. The consumer surplus from a service contract $(\ell,h)$ with price $\pi$ is
\begin{equation}
w(\ell,h,\pi)=U(\ell,h)-\pi \ell . 
\label{welfare} 
\end{equation}
 The maximum price a consumer is willing to pay for a duration $h$ service while receiving surplus $H$ is 
\begin{align} \label{34}
&\pi(h,H) \notag \\
&= \left \{
\begin{array}{ll}
\max \{\pi \ge 0 ~|~ \mbox{ there exists $\ell > 0$ with } w(\ell, h, \pi) \ge H \} \\
\mbox{undefined if there is no $\ell > 0,\pi \ge 0$ with } w(\ell,h, \pi ) \ge H
\end{array}
\right . ,
\end{align}
while the associated power demand is
\begin{equation}
\ell(h,H) = \mathrm{arg} \max_{\ell >0} U(\ell,h) - \pi(h,H)\ell .
\label{maxl}
\end{equation}
For a fixed $h$, $\pi(h,H)$ is defined over a set of values of $H$ of the form $H_{max}(h) \geq H \geq 0$.

From \eqref{welfare}, $h \mapsto w(\ell, h, \pi)$ is increasing in $h$.  So, for a fixed $H$, $\pi (h, H)$ is defined on a set of the form
$T \ge h \ge \hm (H)$.  $ \hm (H)$ is the smallest duration at which a consumer can receive surplus $H$.  

We obtain some structural characteristics of the willingness to pay $\pi(h,H)$.

\begin{lemma} \label{L2}
(1) $\pi(h,H)$ is  strictly increasing in $h$ for $h \ge \hm (H)$ and has non-decreasing increments. (2)  for $H \leq H_{max}(h)$, $\pi(h,H)$ is strictly decreasing in $H$ and  $\ell(h,H)$ is strictly increasing in $H$.
 \end{lemma}
\begin{proof}  1. Since $U(\ell,h)$ is strictly increasing in $h$, the willingness to pay is also strictly increasing in $h$.  From \eqref{welfare}, \eqref{34} and \eqref{maxl} it follows that at $\ell^* = \ell(h,H)$
\begin{equation}
U(\ell^*,h) - \pi(h,H)\ell^* = H.
\end{equation}
So,
\begin{equation}\label{eq:an1}
\pi(h,H) = \frac{U(\ell^*,h) -H}{\ell^*}.
\end{equation}
Also, for any $\ell >0$, $\pi =  \frac{U(\ell,h) -H}{\ell}$ gives a surplus of $H$. Therefore, it follows that 
\begin{equation}\label{eq:an2}
\pi(h,H)  =  \max_{\ell>0} \frac{U(\ell,h)-H}{\ell}
\end{equation}
Using \eqref{eq:an2} and the non-decreasing increment property of the utility function, we have
\begin{align}
\pi(h,H) & =  \max_{\ell} \frac{U(\ell,h)-H}{\ell} \notag \\
 &\leq  \max_{\ell} \frac{\frac{U(\ell,h+1)+U(\ell,h-1)}{2}-H}{\ell} \notag \\
 &\leq \frac{\max_{\ell} \frac{U(\ell,h+1)-H}{\ell}}{2} + \frac{\max_{\ell} \frac{U(\ell,h-1)-H}{\ell}}{2} \notag \\
 &= \frac{\pi(h+1,H) + \pi(h-1,H)}{2}
\end{align}
which implies $\pi(h,H)$ has non-decreasing increments.  

2. From \eqref{eq:an2}, it follows that $\pi(h,H)$ is strictly decreasing in $H$. From \eqref{maxl},
\[\frac{\partial  U}{ \partial \ell} (\ell (h,H), h) = \pi (h, H).\]
Since $U$ is strictly concave in $\ell$, $\frac{\partial  U}{ \partial \ell} (\ell , H)$ is srictly decreasing in $\ell$ and since $\pi (h,H)$ is strictly
decreasing in $H$, $\ell(h,H)$ is strictly increasing in $H$.
\end{proof}

We now present a scheme to construct the equilibrium.  Suppose the available power supply profile is $p_1 \ge p_2 \cdots \ge p_T$.

\textit{Step 1} \
Pick a trial equilibrium consumer surplus $H$.  Use \eqref{34}-\eqref{maxl} to   construct $(T-\hm +1)$ different contracts $(\ell (h, H), h)$ of duration $h=T, \cdots, \hm$, at price $\pi (h,H)$.

\textit{Step 2} \
Assign contract $(\ell (T, H), T)$ to a group of consumers of size $n(T) = p_T / \ell(T,H)$.  Each of these consumers will receive $\ell (T,H)$kW
for the $T$ slots $T, \cdots, 1$.

Assign contract $(\ell( (T-1),H),T-1)$ to a (disjoint) group of consumers of size $n(T-1) = [p_{T-1} -p_T]/\ell(T-1, H)$.  Each of these consumers will receive $\ell (T-1,H)$kW
for the $(T-1)$ slots $(T-1), \cdots, 1$.

Proceed in this way to assign the last contract $(\ell(\hm, H)$ to a (disjoint) group of consumers $n(\hm) = [p_{\hm -1} - p_{\hm }]
/\ell (\hm, H)$.

\textit{Step 3}  \ This allocation serves $N(H) = n(T) + \cdots + n(\hm)$ consumers.  If $N(H) = 1$, all consumers are served and this
is the equilibrium allocation\footnote{The non-decreasing increment property of $\pi(h,H)$ implies that supplier revenue is maximized at these contracts.}.  If $N(H) > 1$ ($< 1$), one must increase (decrease) $H$ and return to Step 1.  From Lemma \ref{L2},
we know that $\partial N(H) /\partial (H) $ is strictly negative, so there is a unique $H^*$ with $N(H^*) = 1$.

The above allocation is illustrated in the  left plot of Fig. \ref{fig-ex}. There are five contracts of duration $h = 5, \cdots, 1$ and
power level $\ell(h, H) $ supplying $n(h) = [p_{h}-p_{h+1}]/\ell(h, H) $ consumers.  The equilibrium surplus $H$ is determined by the market
clearing condition: $N(H)=n(5) + \cdots + n(1) = 1$.
%
\subsection{$U(\ell,h)$ is strictly convex in $\ell$ for all $h>0$ and has non-increasing increments in $h$ for all $\ell>0$.} \label{sec:4b}
We relax the social welfare problem keeping only the total energy as constraint:
\begin{align}
    &\max_{\ell(x),h(x)} \int_0^1 U(x,\ell(x),h(x))dx \notag \\
    & \mbox{subject to} \notag \\
    &\sum_{i=1}^T d_i = \sum_{i=1}^T p_i ,\notag\\
    &d_i := \int_0^1 \ell(x)\ind{\{i \leq h(x)\}}dx, \quad i=1,2,\ldots,T, \notag \\
    & \ell(x) \geq 0,\;\;h(x) \in \{1,2,\ldots,T \}.
\end{align}
Consider any $(\ell(x), h(x))$ and set $E(x) = \ell(x)h(x)$. Since $U(\ell, h)$ is strictly convex in $\ell$ and has non-increasing increments in $h$, 
\begin{equation}
U(E(x),1) \geq h(x) U(\ell(x),1) \geq U(\ell(x),h(x)).
\end{equation}
Thus the optimal allocation of the relaxed problem will have $h(x)=1$. It is straightforward to show that the solution of the relaxed problem is feasible for the original problem. Consequently, in this case, the market outcome  favors contracts of shorter duration reflecting the non-increasing increments on the utility function for $h$, so shorter contracts are more valuable. The  allocation is illustrated by considering a power supply with five levels as in the right plot of Fig. \ref{fig-ex}.

\begin{figure*}[!t]
\centering
\includegraphics[width=4.0in]{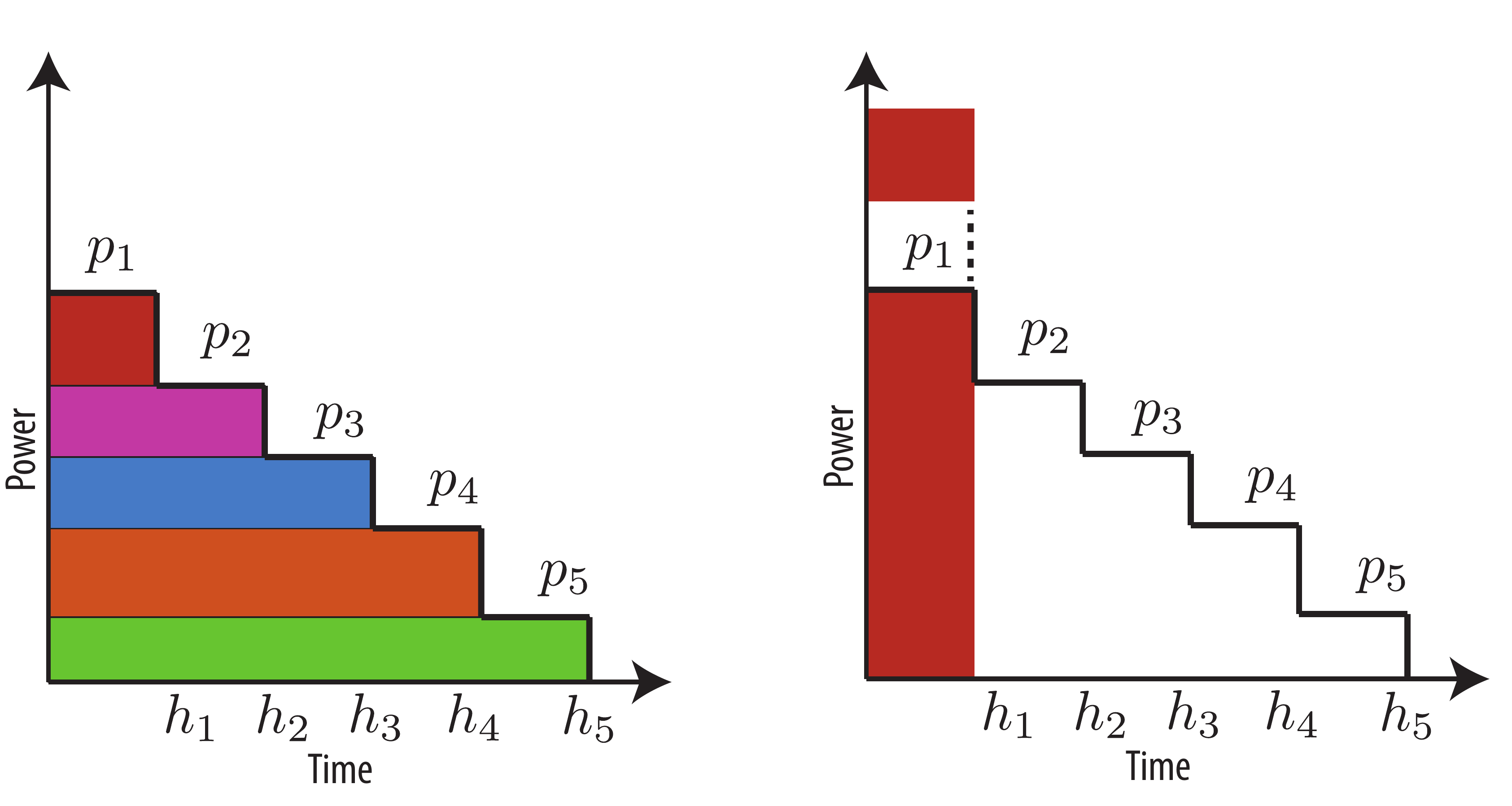}
\caption{Optimal allocations for different forms of the utility functions.}
\label{fig-ex}
\end{figure*}
\begin{remark}{ In section \ref{sec:4a}, for the same power level, the marginal price of the $h$th slot, $\pi(h) - \pi(h-1)$ is \textit{increasing}.  This is due to the fact that from a given 
supply profile it may be possible to produce  a $m$-slot and a $k$-slot service but not a $(m+k)$-slot service, as can be seen from the definition of adequacy.  This contrasts with the assumption in  \cite{Chao-Oren86} that, for conventional generator technology, the marginal price of producing electricity for duration $h$ decreases with $h$. In section \ref{sec:4b}, the  above this issue does not arise since service of a single duration is produced.}
\end{remark}

\section{Spot vs duration-differentiated market} \label{sec-compare}
Today most consumers pay a fixed price for electrical energy that is independent of when and how much they consume, so they
have neither the information nor the incentive to adjust their consumption to changes in supply.  Consequently, many advocate that consumers should face a real-time or spot price that matches the cost of supply to promote efficiency.  We compare the equilibria in a spot market (s) and a duration-differentiated market (d) in a simple example. The results illustrate possible inefficiencies in the market outcomes of a spot pricing scheme with respect to the forward market for  duration-differentiated services.


The supplier receives free renewable power given by the iid sequence $p_t, \ t=1, 2, \cdots, T$.  At the beginning of slot $t$ the supplier learns $p_t$ and offers to sell power in that slot at the spot price of $\pi(t)$ per kW.  The supplier is obligated to sell as much power as is demanded. Any shortfall is made up by purchasing `grid' power at price $C$.  Any surplus is discarded.

Consumers require only unit duration services (that is, $h=1$).  They are indifferent as to when they receive power, which they all evaluate according to the  utility
function $U(\ell)$.  Consumers have a priori expectations $\hat{\pi}(t)$ about the prices $\pi(t)$, and given their flexibility, they will only purchase
electricity in the least costly slot, i.e., at price $\hat{\pi}_s = \min \{\hat{\pi}(t)\}$.  The amount they will purchase is
\[ \hat{\ell} = \ell(\hat{\pi}_s) = \arg \{U_{\ell} (\ell) =  \hat{\pi}_s\} .\] 
The fraction of consumers $n_t \in [0,1]$ that decide to purchase in a particular slot is arbitrary, so we may model the aggregate demand function in slot $t$ as a function of the spot price $\pi (t)$ as
\begin{equation} \label{de}
q_t (\pi (t) ) = \left\{
\begin{array}{ll}
n_t \times \hat{\ell} & \mbox{ if } \pi(t) \le \hat{\pi}_s, \\
0 & \mbox{ if } \pi (t)> \hat{\pi}_s
\end{array}
\right . .
\end{equation}
This demand function has infinite price elasticity.  By contrast, the supply function $s_t(\pi)$ is inelastic: the supplier offers all power for sale, so
\begin{equation} \label{se}
s_t (\pi (t)) = \left\{
\begin{array}{ll}
\infty & \mbox{ if } \pi > C \\
p_t  & \mbox{ if } \pi \le C
\end{array}
\right . .
\end{equation}
From \eqref{de} and \eqref{se}, in equilibrium $\pi(t) = \hat{\pi}_s$ for all $t$, and the supplier must set $\hat{\pi}_s$.
Suppose it costs the supplier $c$ per kW to distribute electricity to the consumer.  The supplier's expected profit over $T$ slots is
\begin{align*}
R_s(\hat{\pi}_s) &= \expec \sum_t [(\hat{\pi}_s -c) n_t \times \hat{\ell} -C (n_t \times \hat{\ell}-p_t)^+] \notag \\
&= (\hat{\pi}_s -c) \hat{\ell} -C \expec \big [\sum_t (n_t \times \hat{\ell}-p_t)^+ \big ],
\end{align*}
since $\sum n_t =1$, so $\hat{\ell}$ is the total amount of energy that consumers  demand at price $\hat{\pi}_s$. The price $\hat{\pi}_s$ that the supplier  sets  depends upon competition in supply.  We assume monopolistic competition where suppliers set their own prices but  which drives expected profit to zero, so
$\hat{\pi}_s$ is the solution of
\begin{equation}
R_s(\hat{\pi}_s) =(\hat{\pi}_s -c)  \ell(\hat{\pi}_s)-C \expec \big[\sum_t (n_t \times \ell(\hat{\pi}_s)-p_t)^+ \big]= 0.\label{sr}
\end{equation}
In the DD market, the supply function of kW-slots is
\begin{equation} \notag
s (\pi ) = \left\{
\begin{array}{ll}
\infty & \mbox{ if } \pi > C \\
\hat{E}  & \mbox{ if } \pi \le C
\end{array}
\right . .
\end{equation}
where $\hat{E}$ is the expected value of the total power $E = \sum p_t$ that the supplier expects to receive, whereas the demand function is simply 
\[q(\pi) = \ell (\pi).\]
So if the supplier sets  price $\hat{\pi}_d$ per kW-slot, the zero expected profit condition   is 
\begin{equation}
R_d(\hat{\pi}_d) = (\hat{\pi}_d -c) {\ell} (\hat{\pi}_d) -C \expec ({\ell}(\hat{\pi}_d)-\sum_t p_t)^+ = 0. \label{dr}
\end{equation}
Comparing the second terms in \eqref{sr} and \eqref{dr} we see, from the convexity of $x \mapsto x^+$, that 
\[ \sum_t (n_t \times \ell -p_t)^+ \ge  ({\ell}-\sum_t p_t)^+,\]
with strict inequality if $(n_t \times \ell -p_t)$ assumes positive and negative values for different $t$, which will likely be the case if
$p_t$ is not constant.  In this case we conclude that
\[\hat{\pi}_s > \hat{\pi}_d \mbox{  and } \ell(\hat{\pi}_s) < \ell (\hat{\pi}_d),\]
so the spot market is leads to a lower consumer surplus than the duration-differentiated market.  The underlying reason for the welfare loss is that in the spot market consumers are less able to reveal their demand flexibility in ways that suppliers can use to match variability in supply.
\begin{remark}
Providing discounts to flexible consumers is not a novel practice. In fact, there already exist other markets (for example, online airfare and hotel reservations) where demand flexibility is used to give discounts to more flexible consumers.
\end{remark}
\section{Conclusions} \label{sec-conclusions}


The use of flexible loads to shape demand in response to supply variation can be an effective alternative to greater reserves to integrate renewable sources. In this paper, we focus on a stylized model of a continuum of flexible loads, each requiring a constant power level for a specified duration within an operational period. The flexibility resides in the fact that the power delivery may occur at any subset of the total period. 
A key characteristic of flexible loads is the fact that many supply profiles can effectively provide a set of these loads. A complete characterization of supply adequacy conditions and an algorithm for allocating supply to meet these loads in an efficient way were developed. 
A forward market for these loads was also investigated. The centralized solution that maximizes social welfare along with the characterization of an efficient competitive equilibrium were presented. These results can be extended in several ways. In particular, the consideration of additional load flexibility constraints, the impact of distribution system constraints and the practical implementation of duration-differentiated (and other flexible demand) markets need further study.

\appendix
\section{Proof of Theorem \ref{thm:continuum_allocation}} \label{sec:allocation_proof}

Suppose $\VEC p$ is simply adequate for  loads $(\ell(x),h(x)), x \in [0,1]$.
Let $\VEC q$ be the supply time profile corresponding to  $\VEC p$. Simple adequacy of $\vec p$ implies that there exists an allocation $B(x,t)$ such that 
\begin{align}
\int_0^1 \ell(x)&B(x,t)dx \leq  q_t ,\\
\sum_{t=1}^T &B(x,t) = h(x) .
\end{align}

We will group together consumers who are being served at the same time slots under the allocation $B(\cdot,\cdot)$ as follows: For each  subset $S \subset \{1,2,\ldots,T\}$, define
\begin{equation}
P_{S} = \{ x \in [0,1] | B(x,t)=1, \forall  t \in S  \mbox{and} B(x,t)=0, \forall t \notin S\}
\end{equation}
The finite collection of sets $P_S, S \subset \{1,\ldots,T\},$ forms a partition of the set of consumers.

The set of consumers being served at time $1$ under $B(\cdot,\cdot)$ is
\[ \mathcal{B}_1 = \bigcup_{S \ni 1} P_S \]

Let $\mathcal{A}_1$ be the set of consumers being served under the LLDF allocation at time $1$.

For any  set $A$, $\mu(A) := \int_{A} \ell(x)dx$ \footnote{$\mu$ can be viewed as a measure on the unit interval. Assuming that $\ell(x)>0$, $\mu$ and the Lebesgue measure are absolutely continuous with respect to each other. Therefore,  an event is $\mu$-almost sure $\iff$ it is Lebesgue almost sure.}. It is easy to verify that 
\begin{equation}
\mu(\mathcal{A}_1) = \min(q_1, \mu([0,1])),
\end{equation}
and that $\mu(\mathcal{B}_1) \leq \mu(\mathcal{A}_1)$.

1. Consider first the case where  $\mu(\mathcal{A}_1) = \mu(\mathcal{B}_1)$. 
Define $C = \mathcal{A}_1\setminus \mathcal{B}_1$ and $D = \mathcal{B}_1\setminus\mathcal{A}_1$.
We must have:
$\mu(C) = \mu(D)$.

Suppose $\mu(C)= \mu(D) >0$. Then, because $P_S$ from a partition of the set of consumers, we can write

\[ C = \bigcup_{S} C \cap P_S \]
\[ D= \bigcup_S D \cap P_S\]

We will refer to the sets $C \cap P_S$ as the ``atoms'' of $C$ and the sets $D \cap P_S$ as "atoms" of $D$.

Because $\mu(C)  = \mu(D) >0$, there must exist sets $S,U \subset \{1,\ldots,T\}$ such that $\mu(C \cap P_S) >0$ and $\mu(D \cap P_U) >0$. 
We need to consider two possibilities:

\begin{enumerate}
\item $\mu(C \cap P_S) \geq \mu(D \cap P_U)$. Find a set $C' \subset C \cap P_S$ such that $\mu(C') = \mu(D \cap P_U)$. Let $D' =D \cap P_U $.
~\\

The basic idea now is to ``swap'' $C'$ and $D'$. 
Every consumer in $C'$ has a longer leftover duration than every consumer in $D'$. Further, all consumers in $C'$ are always served together and all consumers in $D'$ are always served together.  Therefore, there must be a time slot $k$ at which the set $C'$ is being served under allocation rule $B(\cdot,\cdot)$ but $D'$ is not.  
Define the swapped allocation $B'(\cdot,\cdot)$ as being identical to $B(\cdot,\cdot)$ except that 
\[ B'(C',1) =1, ~ B'(C,k) =0,\]
 \[B'(D',k) =1, ~ B'(D',1) =0.\]
It is clear that this operation preserves adequacy.
\item $\mu(C \cap P_S) < \mu(D \cap P_U)$. In this case, let $C' = C \cap P_S$ and find a set $D' \subset D$ with $\mu(D') = \mu(C \cap P_S)$. Repeat the swap as above.

\end{enumerate}

In either of the above cases, after the swap is complete, we have either removed one "atom" of $C$ or one atom of $D$.  That is, if $\mathcal{B}'$ is the set of consumers being served at time $1$ after the swap, then either $\mathcal{A}\setminus \mathcal{B}'$ or $\mathcal{B}' \setminus \mathcal{A}$ has one less atom (of positive measure under $\mu$) than its pre-swap counterpart. Since there are only finitely many such atoms, a finite number of swaps would would transform $\mathcal{B}_1$ to a set that is almost surely identical to $\mathcal{A}_1$ (with respect to the $\mu$ measure). Thus, we have constructed an allocation that agrees ($\mu$ almost surely) with LLDF at time $1$. 

2. In case $\mu(\mathcal{B}_1) < \mu(\mathcal{A}_1)$, first find a subset $\mathcal{E}$ of $\mathcal{A}_1 \setminus \mathcal{B}_1$ with $\mu(\mathcal{E}) = \mu(\mathcal{A}_1) - \mu(\mathcal{B}_1)$. Change $\mathcal{B}_1$ to $\mathcal{B}_1 \cup \mathcal{E}$ and then use the same argument as above.

The argument for future time slots is similar.
\bibliographystyle{IEEEtran}
\bibliography{timedif}

\begin{thebibliography}{10}
\providecommand{\url}[1]{#1}
\csname url@samestyle\endcsname
\providecommand{\newblock}{\relax}
\providecommand{\bibinfo}[2]{#2}
\providecommand{\BIBentrySTDinterwordspacing}{\spaceskip=0pt\relax}
\providecommand{\BIBentryALTinterwordstretchfactor}{4}
\providecommand{\BIBentryALTinterwordspacing}{\spaceskip=\fontdimen2\font plus
\BIBentryALTinterwordstretchfactor\fontdimen3\font minus
  \fontdimen4\font\relax}
\providecommand{\BIBforeignlanguage}[2]{{%
\expandafter\ifx\csname l@#1\endcsname\relax
\typeout{** WARNING: IEEEtran.bst: No hyphenation pattern has been}%
\typeout{** loaded for the language `#1'. Using the pattern for}%
\typeout{** the default language instead.}%
\else
\language=\csname l@#1\endcsname
\fi
#2}}
\providecommand{\BIBdecl}{\relax}
\BIBdecl

\bibitem{CAISO}
CAISO, ``Integration of renewable resources: Operational requirements and
  generation fleet capability at 20\% rps,'' \emph{CAISO Report}, 2010.

\bibitem{kirschen2010}
M.~Ortega-Vazquez and D.~Kirschen, ``Assessing the impact of wind power
  generation on operating costs,'' \emph{Smart Grid, IEEE Transactions on},
  vol.~1, no.~3, pp. 295 --301, dec. 2010.

\bibitem{negwankowshamey12}
M.~Negrete-Pincetic, G.~Wang, A.~Kowli, E.~Shafieepoorfard, and S.~Meyn, ``The
  value of volatile resources in electricity markets,'' \emph{Submitted to IEEE
  Transactions on Automatic Control}, 2013.

\bibitem{cal09K}
D.~S. Callaway, ``Tapping the energy storage potential in electric loads to
  deliver load following and regulation, with application to wind energy,''
  \emph{Energy Conversion and Management}, vol.~50, no.~5, p. 1389 Ð 1400,
  2009.

\bibitem{galus2010}
M.~D. Galus, R.~La~Fauci, and G.~Andersson, ``Investigating {PHEV} wind
  balancing capabilities using heuristics and model predictive control,'' in
  \emph{Power and Energy Society General Meeting, 2010 IEEE}, 2010, pp. 1--8.

\bibitem{papaoren2010}
A.~Papavasiliou and S.~Oren, ``Supplying renewable energy to deferrable loads:
  Algorithms and economic analysis,'' in \emph{Power and Energy Society General
  Meeting, 2010 IEEE}, 2010, pp. 1--8.

\bibitem{matdyscal12K}
J.~L. Mathieu, M.~Dyson, and D.~S. Callaway, ``Using residential loads for fast
  demand response: The potential resource and revenues, the costs and the
  policy recommendations,'' in \emph{2012 ACEEE Summer Study on Energy
  Efficiency in Buildings}, 2012.

\bibitem{anand2012}
A.~Subramanian, M.~Garcia, A.~Dominguez-Garcia, D.~Callaway, K.~Poolla, and
  P.~Varaiya, ``Real-time scheduling of deferrable electric loads,'' in
  \emph{American Control Conference (ACC), 2012}, 2012, pp. 3643--3650.

\bibitem{Paschalidis}
I.~C. Paschalidis, L.~Binbin, and M.~C. Caramanis, ``Demand-side management for
  regulation service provisioning through internal pricing,'' \emph{IEEE
  Transactions on Power Systems}, vol.~27, no.~3, pp. 1531--1539, Aug 2012.

\bibitem{tanvar93}
C.-W. Tan and P.~Varaiya, ``Interruptible electric power service contracts,''
  \emph{Journal of Economic Dynamics and Control}, vol.~17, no.~3, pp. 495 --
  517, 1993.

\bibitem{pravin2011}
P.~Varaiya, F.~Wu, and J.~Bialek, ``Smart operation of smart grid:
  Risk-limiting dispatch,'' \emph{Proceedings of the IEEE}, vol.~99, no.~1, pp.
  40--57, 2011.

\bibitem{negmey12}
M.~Negrete-Pincetic and S.~Meyn, ``{Markets for Differentiated Electric Power
  Products in a Smart Grid Environment},'' in \emph{IEEE PES 12: Power Energy
  Society General Meeting}, 2012.

\bibitem{bitarlow2012}
E.~Bitar and S.~Low, ``Deadline differentiated pricing of deferrable electric
  power service,'' in \emph{Decision and Control (CDC), 2012 IEEE 51st Annual
  Conference on}, 2012, pp. 4991--4997.

\bibitem{dd2013}
A.~Nayyar, M.~Negrete-Pincetic, K.~Poolla, and P.~Varaiya, ``Rate constrtained
  energy services,'' \emph{Submitted to IEEE Transactions on Automatic
  Control}, 2014.

\bibitem{arnold}
B.~C. Arnold, \emph{Majorization and the Lorenz Order: A Brief
  Introduction}.\hskip 1em plus 0.5em minus 0.4em\relax Springer-Verlag, 1987.

\bibitem{aumann65}
R.~Aumann, ``Integrals of set-valued functions,'' \emph{Journal of Mathematical
  Analysis and Applications}, vol.~12, no.~1, pp. 1--12, 1965.

\bibitem{aumann2}
\BIBentryALTinterwordspacing
R.~J. Aumann, ``\BIBforeignlanguage{English}{Existence of competitive
  equilibria in markets with a continuum of traders},''
  \emph{\BIBforeignlanguage{English}{Econometrica}}, vol.~34, no.~1, pp. 1--17,
  1966. [Online]. Available: \url{http://www.jstor.org/stable/1909854}
\BIBentrySTDinterwordspacing

\bibitem{Chao-Oren86}
H.~Chao, S.~Oren, and a.~R.~W. S.A.~Smith, ``Multilevel demand subscription
  pricing for electric power,'' \emph{Energy Economics}, vol.~4, pp. 199--217,
  1986.

\end{thebibliography}

\end{document}